\documentclass[runningheads]{llncs}
\usepackage[textwidth=3cm,tickmarkheight=3pt]{todonotes}
\usepackage[utf8]{inputenc}
\usepackage{amssymb, amsmath}
\usepackage{numprint}
\usepackage{algorithm}
\usepackage[noend]{algpseudocode}
\usepackage{color}
\definecolor{citeblue}{rgb}{0.1,0,.4}
\usepackage[pdftex%
,colorlinks=true%
,bookmarks=true%
,linkcolor=citeblue%
,citecolor=citeblue%
,urlcolor=blue%
,plainpages=false]{hyperref}

\usepackage[nameinlink,capitalize,noabbrev]{cleveref} %
\usepackage{commath}
\usepackage[most]{tcolorbox}
\usepackage{enumitem}
\usepackage{orcidlink}
\usepackage{graphicx}
\usepackage{caption, subcaption}

\newcommand{\qedhere}{$\blacksquare$}
\newcommand{\numvar}[1]{|#1|}
\newcommand{\domain}{\mathbb{D}}
\newcommand{\codomain}{\mathbb{D}'}
\newcommand{\cmpsign}[1]{\lesseqgtr^{#1}}
\newcommand{\cmpsignstrict}[1]{\lessgtr^{#1}}
\newcommand{\coeffsign}[1]{\sigma(#1)}
\newcommand{\ivbound}[2]{\text{bd}(#1,#2)}

\newcommand{\testlength}[1]{}
\usepackage{etoolbox}

\npthousandsep{,}

\newcommand{\smt}{\textsc{smt}}
\newcommand{\args}{\text{args}}
\usepackage{xspace}
\newcommand{\req}{\emph{ReqPivot}\xspace}
\newcommand{\subst}[3]{#1[#2\leftarrow #3]}
\newcommand{\repr}[1]{\hat{#1}}

\newcommand{\todoSN}[1]{\todo[color=orange!50]{#1 -SN}}
\title{Inductive Satisfiability Certification for Universal Quantifiers and\\Uninterpreted Function Symbols}
\titlerunning{Inductive Satisfiability Certification}

\author{Stefan Ratschan\inst{1}\orcidlink{0000-0003-1710-1513} \and
Anggha Nugraha\inst{1,2}\orcidlink{0000-0002-7139-4487} \and
Mikol\'a\v{s} Janota\inst{3}\orcidlink{0000-0003-3487-784X} \and 
Marek Dan\v{c}o\inst{3}\orcidlink{0009-0008-3031-113X}}
\authorrunning{S. Ratschan et al.}
\institute{Institute of Computer Science, The Czech Academy of Sciences,
Prague, Czech Republic \\
\email{stefan.ratschan@cs.cas.cz} \and 
Faculty of Information Technology, Czech Technical University in Prague,
Prague, Czech Republic \and 
Czech Institute of Informatics, Robotics and Cybernetics, Czech Technical University in Prague, Czech Republic}

\begin{document}
\maketitle

\begin{abstract}
  The combination of uninterpreted function symbols and universal quantification occurs in many applications of automated reasoning, for example, due to their ability to reason about arrays. Yet the satisfiability of such formulas is, in general, undecidable. In practice, SMT solvers are often successful in the unsatisfiable case, using heuristics. However, in the satisfiable case, they rely on  explicit model construction, which fails for formulas whose smallest model is not small enough. We introduce an alternative approach that certifies satisfiability using induction arguments, and apply it to the case of linear integer arithmetic. The resulting algorithm is able to prove satisfiability of formulas that are out of reach for current SMT solvers. 

\end{abstract}

\section{Introduction}

Current $\smt$ solvers have mainly been designed to prove unsatisfiability. As a consequence,  they have limited capabilities to detect satisfiability of formulas not belonging to known decidable classes. However, the ability to detect satisfiability is important for identifying counter-examples, which is indispensable in many applications~\cite{clarke25}.

In this paper, we address this problem for formulas containing universal quantifiers and uninterpreted function symbols. We mainly concentrate on the combination with linear arithmetic, which corresponds to the SMTLIB logic UFLIA, which is undecidable~\cite{Downey:72,Horbach:17}. However, such formulas occur in important applications, for example, in program verification, where uninterpreted function symbols play an important role in reasoning about arrays~\cite{bradley-manna07}. Current $\smt$ solvers mainly rely on heuristics~\cite{Reynolds:17} to prove unsatisfiability of such formulas. To prove satisfiability, they explicitly construct a model~\cite{ge-cav09,reynolds2013finite}. This fails for formulas as simple as $f(0)=0\wedge \forall x \;.\; f(x+1)=f(x)+1$, which do not have a finite model, and is costly for formulas whose smallest model is finite yet large.

In this paper, we explore an alternative approach: instead of explicit model construction, we prove satisfiability using inductive arguments. More specifically, we
\begin{itemize}
\item introduce a notion of satisfiability certificate that represents inductive arguments for satisfiability of such formulas,
\item present an algorithm that proves satisfiability by synthesizing such a satisfiability certificate for a sub-class of UFLIA formulas,
\item prove the algorithm's soundness and completeness under the assumption that the input formula satisfies a certain semantic condition enabling inductive reasoning,
\item provide an $\smt$ encoding for formulas solved by this algorithm, and 
\item demonstrate through computational experiments that our approach can prove satisfiability of formulas beyond the reach for current $\smt$ solvers.
\end{itemize}

\section{Related Work}

Our work intersects several areas of automated reasoning: model-based quantifier instantiation, decision
procedures for arrays, synthesis, and proof certificates.

\paragraph{Model-based and instantiation methods.}
Reynolds et al.~\cite{reynolds2016model,reynolds2013finite} develop finite model finding techniques for quantified formulas. Their method requires users to ensure admissibility and may introduce many case splits (``merge lemmas'') to satisfy cardinality constraints. Ge and de Moura~\cite{ge-cav09} introduce Model-Based Quantifier Instantiation (MBQI), noting that ``For certain cases containing offsets on array indices, our procedure will result in an infinite set of instantiations, while a decision procedure of LIA
will terminate. It leaves an open question for future research on how to reason about this type of formulas within an instantiation-based approach.'' Our work attempts to  address this open question.

\paragraph{Decision procedures for arrays and quantified formulas.}
Bradley et al.~\cite{bradley-vmcai06} present a decision procedure for array properties but explicitly forbid arithmetic on the universally quantified variables. Habermehl et al.~\cite{habermehl-etaps08} allow constraints of the form $a[i+c]$ with difference bounds, but remain limited to offset-bounded indices. Alberti et al.~\cite{alberti2014decision} are restricted to flat formulas with one universally quantified variable. This excludes terms such as $a[i+1]$ that are the main focus of our paper. Elad et al.~\cite{elad2024infinite} prove decidability for a first‑order fragment without arithmetic, unable to express our constraints.

\paragraph{Synthesis.}
\emph{(Program) synthesis} has been studied extensively. Techniques typically rely on
enumeration of candidates, such as counterexample-guided approaches
(CEGIS~\cite{cegis}), where some support recursion~\cite{Farzan2021,Kneuss2013}.
Syntax-Guided Synthesis (SyGuS)~\cite{sygus} provides a unifying framework:
given a context-free grammar and a logical specification, find a conforming
expression. Its performance depends crucially on the chosen grammar. SyGuS has
been extensively studied in the context of
SMT~\cite{reynolds-fmsd19,abate-jar23,ReynoldsBLT20}.
Model finding can be phrased as a synthesis problem~\cite{ParsertBJK23}.

Synthesis has also been tackled by providing \emph{templates} for the objects to
be synthesized, which has also been used in the context of model
finding~\cite{BrownJO24,elad2024infinite,BrownCJOR25}. 

Extended study has also been carried out on specifications that observe the
\emph{single-invocation property}, which means that the function to be
synthesized is always applied to the same exact list of arguments. This means
that, effectively, the function can be deskolemized and quantifier elimination
procedures can be adapted to  decide the
witness~\cite{kuncak-pldi10,kuncak-sttt13,hozzova-fmcad25}.  In this context,
Hozzová et al.~\cite{hozzova-ijcar24} synthesize recursive functions from proofs
with induction. Our approach also works outside of the single-invocation
fragment, e.g.\
$\forall x.\,f(x)+f(x)>x$ is deskolemizable but $\forall x.\,f(x)+f(x+1)>1$ is not.

\paragraph{Proof certificates and model representation.}
Representation and certification of refutation proofs and models is a
long-standing and ongoing challenge in the automated reasoning communities.
There have been recent efforts in TPTP~\cite{SteenSFM23} and SMT~\cite{Barbosa:23,BarbosaROTB19,lachnitt_et_al:LIPIcs.ITP.2025.26,leansmt}
facing both technical but also theoretical challenges~\cite{BrownJK22} (recall
that interpretations of FOL formulas are not recursively enumerable). In certain theories (e.g., real numbers with transcendental function symbols~\cite{Ratschan:23}), certifying satisfiability may also be a non-trivial theory-specific problem.

\section{Problem Statement}
\label{sec:problem-statement}

We introduce a method for certifying satisfiability of $\smt$ formulas with uninterpreted function symbols whose domain is a sort associated with a theory that has a countable canonical model. In the case where this theory is linear integer arithmetic, the resulting formulas belong to the 
SMTLIB theory UFLIA. The formulas we consider are
of the form
\begin{equation}\label{eq:main-fragment}
    F \land \forall \bar{x}.\; Q%
\end{equation}
where $F$ and $Q$ are quantifier‑free,  $\forall \bar{x}$ stands for universal quantification of arbitrarily many variables, and where  the arguments of all uninterpreted function symbols contain no uninterpreted function symbols of non-zero arity (i.e., the arguments may only contain uninterpreted constants).  

Based on the general method for satisfiability certification, we will design an algorithm for detecting satisfiability in UFLIA, under the following additional assumptions:
\begin{itemize}
    \item All uninterpreted function symbols have maximal arity $1$.  
    \item The number of quantified variables is $1$.
    \item The coefficient of the quantified variable is uniform across all argument terms of any given uninterpreted function symbol in the quantified part~$Q$.
    \end{itemize}

The algorithm is guaranteed to terminate under an additional semantic condition that enables induction (\cref{def:reqpivot-param}). We illustrate the class of formulas our algorithm can handle using two examples. 
\begin{example}
  \label{ex:formula}
  \[f(0)=0 \wedge \forall x \,.\, f(-2x+3)=g(x)+f(-2x+c+1)-c,\]
  where $c$ is an uninterpreted constant, and the coefficient of the quantified variable~$x$ is $-2$ for all arguments of $f$, and $1$ for all arguments of $g$.
\end{example}

\begin{example}[Invalid]
  \[\forall x \,.\, f(2x+3) = f(x+1)\]
  This formula is \emph{not} in our fragment because $f$ appears with arguments having different coefficients ($2$ and $1$).
\end{example}

Note that certain formulas with several quantified variables can be easily transformed to this class by a change of variables. For example, the formula
\[f(0)=0 \wedge \forall y,z \,.\, f(2y-2z+3)=g(-y+z)+f(2y-2z+c+1)-c,\]
can be reduced to Example~\ref{ex:formula} by the change of variables $x=-y+z$.

We are not aware of any result showing the decidability of this class. However, we are also not aware of any result proving it to be undecidable. Indeed, our goal is not decidability classification, but an algorithm for showing satisfiability that allows straightforward implementation in $\smt$ solvers, and that is efficient in practice.

\section{Notation and Terminology}

As usual in an $\smt$ context, we will work with a first-order language that is sorted. We also assume that variables occur only in positions where they are bound by a quantifier. We have already used the notation $\forall \bar{x}$ as an abbreviation for universal quantification over several variables. In this case, we will write $\numvar{\bar{x}}$ for the number of those variables. For a formula $F$ and uninterpreted function symbol $f$, $args(F, f)$ will denote the set of argument terms of $f$ in $F$. We will write $\subst{s}{v}{t}$ for the result of substituting the term $t$ for $v$ in the term or formula~$s$. Here, $v$ may again represent several variables, in which case we require $t$ to represent a $\numvar{v}$-vector of terms, and assume substitution to be parallel. 

We will sometimes work with the signs from the set $\{+,-\}$, where
for any two signs $s$ and $s'$, $ss'$ is $+$ iff $s=s'$, and $-$ otherwise. We will use such signs to switch between greater and less than relations in the way that for two integers~$z$ and $z'$, $z\cmpsignstrict{+} z'$ iff $z<z'$ and $z\cmpsignstrict{-} z'$ iff $z>z'$. For an interval $B=[b_{\min}, b_{\max}]$, we define $\ivbound{B}{-}= b_{\min}$ and $\ivbound{B}{+}= b_{\max}$.

We use the classical mathematical notion of a well-order. A well-order on a set $S$ is a total ordering on $S$
such that every non-empty subset of $S$ has the least element. Any well-order allows well-founded induction; to prove
a property $P$ holds for all $x$, it suffices to prove that for every $x\in S$, if for all $y\prec x$, $P(y)$ then also $P(x)$.

We use the notion of certificate in the following sense~\cite{McConnell:11}: Given a property~$P(x)$ that is typically difficult to check, a certificate for this property is an object $\Pi$ for which we have a property $P'(x,\Pi)$ and a proof that $P'(x,\Pi)$ implies $P(x)$, and for which checking $P'(x,\Pi)$ is easier than $P(x)$. Here, we will not always insist on a precise meaning of the notion of being easy to check. However, the final goal is a certificate $\Pi$, for which $P'(x,\Pi)$ can be checked more efficiently than the original property $P(x)$, using an algorithm that is in some sense simpler than an algorithm for checking $P(x)$. As a consequence, such a certificate will enable an efficient independent check of the correctness of the result of an algorithm for checking~$P(x)$.

\section{Satisfiability Certificates}
\label{sec:certificate}

In this section, we define objects that certify that a formula of the form $F \land \forall \bar{x}.\; Q$ is satisfiable. We assume that the quantified variables $\bar{x}$ have a sort $\domain$, and that all uninterpreted function symbols have domain~$\domain$ and codomain~$\codomain$. We also assume that both $\domain$ and $\codomain$ have a canonical model of the same name, with $\domain$ being countable.
Finally, we assume that for every element $z\in\domain$, there is a corresponding interpreted constant $\repr{z}$ in the language that denotes~$z$.

In the next section, we will then introduce an algorithm to compute such a certificate in the case where both $\domain$ and $\codomain$ are the integers (more specifically, formulas from UFLIA), and under the additional restrictions described in the problem statement. Also in the present section, we use the integers for illustration in all examples. We will use all semantic notions (e.g., the symbol $\models$) relative to the underlying theory, which is the combination of uninterpreted function symbols, $\domain$ and $\codomain$, in general, and UFLIA, for all examples.

Traditionally, satisfiability of a logical formula is certified by an interpretation that assigns a value to each non-logical symbol in such a way that the formula is true. However, the presence of uninterpreted function symbols with infinite domain raises the question of how to represent such an interpretation.
Our approach assigns values not to all elements of the infinite domain of uninterpreted function symbols in one step, but to each individual element of this domain separately:
\begin{definition}[Cell]\label{def:cell}
	A \emph{cell} is a term of the form $f(n_1,\dots, n_k)$ where $f$ is an uninterpreted function symbol and $n_1,\dots, n_k\in \domain$.
\end{definition}
We denote the set of cells by $C(\domain)$.
Examples of cells are $c$ and $f(7)$. The term $f(c+7)$ is not a cell.

\begin{definition}[Cell Interpretation]
	A \emph{cell interpretation} is a partial function in $C(\domain)\rightharpoonup\codomain$. A cell interpretation $I$ is \emph{compatible} with an interpretation $\mathcal{I}$ iff $I$ and $\mathcal{I}$ agree on all cells where $I$ is defined.
\end{definition}

Given a cell interpretation $I$, we denote by $\mathit{def}(I)$ the set of cells for which $I$ is defined. Given a term $t$ and a cell interpretation $I$  we define $I(t):=\mathcal{I}(t)$, if  $\mathcal{I}(t)$ is unique for every interpretation $\mathcal{I}$ compatible with $I$, and otherwise $I(t):=\bot$.
For example, for $I=\{ f(0) \mapsto 7 \}$, $I(f(0))=7$, but $I(f(1))=\bot$. Also note that $I(f(1)-f(1))=0\neq\bot$. In a similar way, for a formula $\phi$, we define that a cell interpretation~$I$ satisfies $\phi$ iff for every interpretation $\mathcal{I}$ compatible with $I$, $\mathcal{I}$ satisfies $\phi$, in which case we also write $I\models \phi$.

In order to allow a computer representation of cell interpretations, we will always work with cell interpretations that are only defined on finitely many cells.
However, we require that enough cells are defined to (i)~check satisfaction of the quantifier-free part of formulas, and (ii)~determine the value of all arguments of uninterpreted function symbols in the universally quantified part for every instantiation of the universal quantifier.

\begin{definition}[Pre-satisfiability Certificate]%
	\label{def:pre-sat-cert}
	Given a formula $\phi$ of the form $F\wedge \forall \bar{x}\;.\; Q$, we call a cell interpretation $I$ such that
	\begin{itemize}
		\item $I\models F$, and
		\item for all $z\in\domain^{\numvar{\bar{x}}}$, for every uninterpreted function symbol~$f$ and its argument terms $t\in args(Q, f)$, $I(\subst{t}{\bar{x}}{\repr{z}})\neq\bot$
	\end{itemize}
	a \emph{pre-satisfiability certificate} of $\phi$.
\end{definition}
For example, $I=\{ c\mapsto 5 \}$ is a pre-satisfiability certificate of $c\geq 5 \wedge \forall x \;.\; f(x+c+1)= f(x+c)+1$.

A pre-satisfiability  certificate interprets all arguments of uninterpreted function symbols in $F$ to a value in the domain $\domain$, and---together with a value for the quantified variables---also in $Q$. This allows us to identify, for each instantiation, the set of cells that become relevant.
\begin{definition}[Relevant Cells]
  For any pre-satisfiability certificate~$I$, and $z\in\domain^{\numvar{\bar{x}}}$, the set of \emph{relevant cells} $\Gamma_{I,z}(F\wedge \forall \bar{x}\;.\; Q):=$
\[\{ f(I(\subst{t}{\bar{x}}{\repr{z}})) \mid t\in args(Q, f), f \emph{ is an uninterpreted function symbol}\}.\]
\end{definition}
Continuing the example, since $I$ sets the value of $c$ to $5$, $\Gamma_{I,10}(\phi)=\{ f(15), f(16)\}$.

Given a pre-satisfiability certificate, we can now certify the satisfiability of formulas by propagating the values given by a pre-satisfiability certificate in such a way that for each value of the universally quantified variables, the corresponding quantified sub-formula holds.
\begin{definition}[Satisfiability Certificate]%
	\label{def:certificate}
	Given a formula $\phi$  of the form $F\wedge \forall \bar{x}\;.\; Q$, a \emph{satisfiability certificate for $\phi$} is
	\begin{itemize}
        \item a pre-satisfiability certificate $I$ of $\phi$,
		\item a well-order $\preceq$ on $\domain^{\numvar{\bar{x}}}$,%
		\item for every $z\in\domain^{\numvar{\bar{x}}}$
                  \begin{itemize}
                  \item a set $X_z\subseteq \Gamma_{I,z}(\phi)$  (the \emph{propagated cells}) such that \[X_z\cap \mathit{def}(I)=\emptyset \text{ and } X_z\cap\bigcup_{z'\prec z} \Gamma_{I,z'}(\phi)=\emptyset,\] and
                  \item a certificate (the \emph{satisfiability propagator}) showing that for all values of cells in $\Gamma_{I,z}(\phi)\setminus \mathit{def}(I)\setminus X_z$ there exist values for the cells in $X_z$ such that for the cell interpretation $I_{X_z}$ assigning these values to $X_z$, $I\cup I_{X_z}\models \subst{Q}{\bar{x}}{\repr{z}}$.
                  \end{itemize}
       \end{itemize}                
\end{definition}

Intuitively, a satisfiability certificate encodes an induction proof, where the well-order~$\preceq$ sets the direction of induction and each step of the induction shows satisfiability of the universally quantified part~$Q$ when enforcing another choice~$z$ for $\bar{x}$. Such a step is allowed to choose values for the propagated cells~$X_z$. These cells must be chosen in such a way that the propagated cells neither have already been given a value by the pre-satisfiability certificate~$I$ nor by an earlier induction step. The satisfiability propagator is only required to show the existence of appropriate values for the propagated cells. In practice, it will usually assign a value to these cells.

\begin{example}\label{ex:certificate-example}
A satisfiability certificate for the formula \[f(3, a) \geq 4 + g(b) \wedge \forall x, y \:.\: f(x, y) < x + g(y) \wedge g(y) = g(y+b) + 2b\] is 
given by the pre-satisfiability certificate $I = \{ a \mapsto 0, b \mapsto 1, f(3,0) \mapsto 4, g(1) \mapsto 0 \}$ and the following table, where the first column indicates the order~$\preceq$:
    \[\small
    \begin{array}{|l|l|l|l|}\hline
      & X_{(x,y)} & propagators\\ \hline
        x=3, y=0  & \{g(0)\} & g(0)\mapsto 2\\
        x=3,y=1,2,\dots & \{f(3,y), g(y+1)\} & g(y+1)\mapsto g(y)-2, f(x,y)\mapsto y+g(y)-1\\
        x=3, y=-1,-2,\dots &  \{f(3,y), g(y)\} & g(y)\mapsto g(y+1)+2, f(x,y)\mapsto y+g(y+1)+1 \\
       x\neq 3, y \text{ arbitrary}  &   \{f(x,y) \} & f(x,y)\mapsto y+g(y)-1\\ \hline
    \end{array}
  \]
\end{example}

The notion of certificate defined here is abstract in the sense that it does not fix a concrete finite representation that allows for an independent check of whether a given object is indeed a satisfiability certificate. For arriving at concrete independently checkable certificates, it suffices to fix a well-order and decide on a representation of propagated cells and satisfiability propagator that enables such an independent check. For example, all these definitions could be written in the language of linear integer arithmetic, which allows the independent check to be done using a corresponding decision procedure. We will illustrate this in the next section for the special class of formulas described in \cref{sec:problem-statement}.

Also note that for $z\in\domain^{\numvar{\bar{x}}}$ such that $\Gamma_{I,z}(\phi)\subseteq \mathit{def}(I)$, the definition requires the set of propagated cells $X_z$ to be empty. As a consequence, for such $z$, the satisfiability propagator only needs to certify that $I\models \subst{Q}{\bar{x}}{\repr{z}}$. In other words, the cell interpretation~$I$ serving as a pre-satisfiability certificate already plays the role of the satisfiability propagator. 

\begin{theorem}\label{thm:certThenSat}
	Every formula $\phi$  of the form $F\wedge \forall \bar{x}\;.\; Q$ that has a satisfiability certificate is satisfiable.
\end{theorem}

\begin{proof}
	Assume a formula of the given form and assume that it  has a satisfiability certificate.
	The corresponding pre-satisfiability certificate $I$ satisfies the quantifier free part~$F$ of $\phi$.
    
    We first show that for every $z\in\domain^{\numvar{\bar{x}}}$, there exists an extension $I'$ of $I$ such that $I'$ satisfies all quantifier instantiations occurring before $z$ wrt.\ the order $\preceq$, that is,  $I'\models\{ \subst{Q}{\bar{x}}{\repr{z}'} \mid z'\preceq z\}$. For this we use induction on well-orders: Let $z\in\domain^{\numvar{\bar{x}}}$, and let $\Phi$ be $\{ \subst{Q}{\bar{x}}{\repr{z}'} \mid z'\prec z\}$. Assume there is an extension $I_{\Phi}$ of $I$ with $I_{\Phi}\models\Phi$. We prove there is an extension $I'$ of $I$ with $I'\models\Phi\cup \{ \subst{Q}{\bar{x}}{\repr{z}}\}$. %
    By \cref{def:certificate}, the propagated cells $X_z$ are disjoint from all cells appearing in instantiations $z' \prec z$, and hence do not occur in any element of $\Phi$. Let $I'$ be the extension of $I_{\Phi}$ that assigns to the propagated cells $X_z$ the values whose existence is guaranteed by the satisfiability propagator. Since these cells do not occur in $\Phi$, $I'$ and $I_{\Phi}$ coincide on all cells occurring in $\Phi$. Hence $I'$ satisfies both $\Phi$ and $\subst{Q}{\bar{x}}{\repr{z}}$, and so it satisfies $\Phi\cup \{ \subst{Q}{\bar{x}}{\repr{z}}\}$.

	Now for every $z\in\domain^{\numvar{\bar{x}}}$, denote by $I_z$ the cell interpretation satisfying $\subst{Q}{\bar{x}}{\repr{z}}$ constructed as above.
    Observe that for every cell~$u$, and for all~$z, z'\in \domain^{\numvar{\bar{x}}}$ with $I_z(u)\neq \bot$ and $I_{z'}(u)\neq \bot$, $I_z(u)=I_{z'}(u)$. Let $I_\phi$ be the (classical) interpretation defined as follows: for any cell $u$, if there exists a $z$ such that $I_z(u)\neq \bot$, then $I_\phi(u)=I_z(u)$; otherwise $I_\phi(u)=0$. Clearly $I_\phi\models F$. Moreover, for every $z\in\domain^{\numvar{\bar{x}}}$,  $I_\phi\models  \subst{Q}{\bar{x}}{\repr{z}}$, since $I_\phi$ coincides with $I_z$ on all cells where $I_z$ is defined. Hence  $I_\phi\models \phi$.
\qedhere\end{proof}

A satisfiability certificate only shows that a formula is satisfiable. It does not provide concrete values of the corresponding model. We now show how to compute such values from a given satisfiability certificate containing a pre-satisfiability certificate $I$. For this, we require the satisfiability propagators to be constructive. So, for $z\in\domain^{\numvar{\bar{x}}}$ and for each $u\in X_z$ we assume a function $prop_{u,z}$ such that for
a cell interpretation $I'$ with $\mathit{def}(I')=\Gamma_{I,z}(\phi)\setminus \mathit{def}(I)\setminus X_z$, $I\cup I'\cup \{ u \mapsto prop_{u,z}(I\cup I') \mid u\in{X_z}\}\models\subst{Q}{\bar{x}}{\repr{z}}$.

Observe that for any cell $u$ there is at most one $z\in\domain^{\numvar{\bar{x}}}$ with $u\in X_z$ (if there were two distinct instantiations, then since $\preceq$ is a well-order, there is a minimal one, call it $z_{\min}$, and then, for any $z'\succ z$, the requirements on $X_{z'}$ exclude the possibility of $X_{z'}$ containing $u$, as well). So for any cell $u$, let $\mathit{inst}(u)$ either be the $z\in\domain^{\numvar{\bar{x}}}$ with $u\in X_z$ or $\bot$ if there is no such instantiation.

Then, for the assignment $I_\phi$ from the proof of Theorem~\ref{thm:certThenSat}, which satisfies a given formula~$\phi$, and a given cell~$u$, the value of $I_\phi(u)$ can be computed by the following recursive function $val(u): C(\domain)\rightarrow \codomain$:
\begin{tabbing}\hspace*{0.4cm}\=\hspace*{0.4cm}\=\kill
	\textbf{if} $I$ assigns a value $x$ to $u$ \textbf{then} \textbf{return} $x$\\
	\textbf{if} $\mathit{inst}(u)=\bot$ \textbf{then} \textbf{return} any element of $\domain$\\
	\textbf{return} $prop_{u,\mathit{inst}(u)}(I\cup \{ v\mapsto val(v) \mid v\in \Gamma_{I,\mathit{inst}(u)}(\phi)\setminus \mathit{def}(I)\setminus X_{\mathit{inst}(u)}\})$
\end{tabbing}

Here, the recursive calls go down the order $\prec$ in the following sense:
\begin{property}
  \label{prop:recmonot}
	Let $u$ be a cell such that $\mathit{inst}(u)\neq\bot$. Then for any other cell~$v\in\Gamma_{I,\mathit{inst}(u)}(\phi)\setminus \mathit{def}(I)\setminus X_{\mathit{inst}(u)}$ with $\mathit{inst}(v)\neq\bot$, we have $\mathit{inst}(v)\prec \mathit{inst}(u)$.
\end{property}

So each recursive call of the program $val$ either results in $\mathit{inst}(u)$ decreasing or being $\bot$. In the latter case, the program terminates immediately, and the former case cannot happen infinitely often, since the used order is well-founded. So we have:
\begin{property}
The program $val$ terminates.
\end{property}

Note that, in general, it is not possible to compute $val(u)$ by computing the values of $X_z, z \in \domain^{\numvar{\bar{x}}}$, with $z$ being initialized with the minimal $z$ wrt. $\prec$, and then using a loop that increases $z$ from one successor wrt. $\prec$ to the next. For example, if  $\prec$  orders all even integers before all odd ones, then this loop would have to iterate over infinitely many even numbers before computing the value of an odd one.

\section{Algorithm}\label{sec:alg}

We now proceed from the  general certification framework to concrete computation. Throughout this section we assume a formula $\phi$ of the form~$F \land \forall \bar{x}.\; Q$, but with the additional restrictions described in the problem statement (\cref{sec:problem-statement}). In particular, we will allow only one quantified variable~$x$. The restrictions also require all arguments of any uninterpreted function symbol $f$ in $Q$ to have the same coefficient, and we will denote the sign of this coefficient by $\coeffsign{f}$.

We introduce an algorithm that checks satisfiability of $\phi$  by computing a satisfiability certificate of a certain form. 
The computed certificate explicitly certifies satisfiability of the quantifier-free part $F$ and satisfiability over a finite interval $B = [b_{\min}, b_{\max}]$ of instantiations of the quantified part $Q$. For values outside $B$, it uses satisfiability propagators with an order that propagates upward from $b_{\max}+1$ and downward from $b_{\min}-1$. We show soundness of the algorithm and completeness relative to conditions that ensure such propagation is possible.

We now formalize the condition ensuring that propagation outward from a base interval $B$ is possible. Let $\mathcal{F}$ be the set of uninterpreted function symbols in the quantified part $Q$ of $\phi$, and let $\mathcal{T}$ be the set of terms of the form $f(t)$ in $Q$. For a propagation direction $s \in \{+, -\}$ (with $+$ denoting increasing $x$ and $-$ decreasing $x$), we select a subset $\mathcal{S} \subseteq \mathcal{T}$ of terms that are \emph{extremal} in the sense that they are all equal and, depending on the sign of the coefficient and the propagation direction, stand in the appropriate strict order to other argument terms of the same function symbol. The \emph{propagability condition} $\Phi_{\text{prop}}(\mathcal{S})$ guarantees that, regardless of what values are assigned to the non‑extremal occurrences, we can find values $v_1,\dots,v_r$ such that assigning $v_i$ to all occurrences in $\mathcal{S}$ makes $Q$ true. Different function symbols may receive different values. The condition for propagation in the opposite direction is symmetric, using a different subset $\mathcal{S}' \subseteq \mathcal{T}$.

\begin{definition}[\req Condition]
\label{def:reqpivot-param}
For any choice of subsets $\mathcal{S}\subseteq\mathcal{T}$ and signs $s\in\{ -, +\}$, define:
\begin{itemize}
\item \textbf{Extremal condition $\Phi^s_{\text{ext}}(\mathcal{S})$:}
  \vspace*{-0.2cm}
  \[
        \subst{ \bigwedge_{f\in\mathcal{F}}\left(\bigwedge_{f(t),f(t') \in \mathcal{S}} t = t' \;\land\;
        \bigwedge_{f(t) \in \mathcal{S}, f(t') \in \mathcal{T}\setminus\mathcal{S}} t' \cmpsignstrict{\coeffsign{f}s} t
        \right)}{x}{\repr{0}}.
    \]
    
    \item \textbf{Propagability condition $\Phi_{\text{prop}}(\mathcal{S})$:}
        \[\forall x \forall \bar{u} \exists \bar{v}.\ Q'(\mathcal{S})\]
        where $\bar{u}$ corresponds to the variables $u_{f(t)}$ with $f(t)\in\mathcal{T}\setminus\mathcal{S}$ and $\bar{v}$ to the variables $v_{f(t)}$ with $f(t)\in \mathcal{S}$, and $Q'$ is obtained from $Q$ by simultaneously replacing
        \begin{itemize}
            \item every $f(t)\in\mathcal{T}\setminus\mathcal{S}$ by $u_{f(t)}$, and
            \item every $f(t)\in\mathcal{S}$ by $v_{f(t)}$.
        \end{itemize}
\end{itemize}
We say that $Q$ satisfies the \emph{\req condition} for subsets $\mathcal{S}\subseteq\mathcal{T}$ and $\mathcal{S}'\subseteq\mathcal{T}$ iff $\Phi^+_{\text{ext}}(\mathcal{S})$, $\Phi_{\text{prop}}(\mathcal{S})$,
$\Phi^{-}_{\text{ext}}(\mathcal{S}')$, and $\Phi_{\text{prop}}(\mathcal{S}')$ are satisfiable. Moreover, $Q$ satisfies the \emph{\req condition} iff such a pair~$(\mathcal{S}, \mathcal{S}')$ exists.
\end{definition}

\begin{example}[\req\ Illustration]
\label{ex:req-illustration}
Consider
\[
\forall x.\;f(x+c+3)+f(x+4)+f(x+1)=g(x+2)+h(x),
\]
with ground constant $c$. All coefficients are $+1$, hence
$\coeffsign{f}=\coeffsign{g}=\coeffsign{h}=+$.
The set of terms is
\[
\mathcal{T} =\{f(x+c+3),\,f(x+4),\,f(x+1),\,g(x+2),\,h(x)\}.
\]
\smallskip
\noindent\textbf{Upward propagation ($s=+$).}
Let $\mathcal{S}=\{f(x+4),\,g(x+2)\}$.
For $f$, since $\coeffsign{f}s=+$ and $\cmpsignstrict{+}$ is $<$, the part of the extremal conditions under the parentheses reads
\[
x+c+3<x+4\;\land\;x+1<x+4,
\]
which, after substituting $0$ for $x$, yields $c<1$. For $g$ and $h$, $\mathcal{T}\setminus\mathcal{S}$ does not contain corresponding terms, so the condition is trivial.

Propagability becomes
\[
\begin{aligned}
    \forall x\;\forall u_{f(x+c+3)},\,u_{f(x+1)},\,u_{h(x)}\;\exists v_{f(x+4)},\,v_{g(x+2)}\;\text{such that} \\
    \qquad
    u_{f(x+c+3)}+v_{f(x+4)}+u_{f(x+1)} = v_{g(x+2)}+u_{h(x)}.
\end{aligned}
\]
This is satisfiable, e.g.\ by choosing $v_{g(x+2)}=0$ and $v_{f(x+4)}=u_{h(x)}-u_{f(x+c+3)}-u_{f(x+1)}$.

\smallskip
\noindent\textbf{Downward propagation ($s=-$).}
Let $\mathcal{S}'=\{f(x+1),\,h(x)\}$. Since $\coeffsign{f}s=-$ and $\cmpsignstrict{-}$ is $>$,  the part of the extremal conditions under the parentheses reads
\[
x+c+3>x+1\;\land\;x+4>x+1,
\]
which, after substituting $0$ for $x$, yields $c>-2$. The condition for $h$ is trivial.

Propagability becomes
\[
\begin{aligned}
    \forall x\;\forall u_{f(x+c+3)},\,u_{f(x+4)},\,u_{g(x+2)}\;\exists v_{f(x+1)},\,v_{h(x)}\;\text{such that} \\
    \qquad
    u_{f(x+c+3)}+u_{f(x+4)}+v_{f(x+1)}=u_{g(x+2)}+v_{h(x)},
\end{aligned}
\]
which is satisfiable, e.g. by choosing $v_{h(x)}=0$ and $v_{f(x+1)}=u_{g(x+2)}-u_{f(x+c+3)}-u_{f(x+4)}$.

\smallskip
\noindent\textbf{Conclusion.}
$Q$ satisfies the \req\ condition for $(\mathcal{S},\mathcal{S}')$
iff $c\in(-2,1)$.
\end{example}

Now, given a concrete interval $B = [b_{\min}, b_{\max}]$, we need additional ``clash conditions'' that prevent the cells generated during propagation from clashing with the ground part $F$.

\begin{definition}[Interval Extension Formulas]
\label{def:interval-extension-formulas-param}
For an interval $B\subseteq\mathbb{Z}$, a sign~$s\in\{-,+\}$, and a subset $\mathcal{S} \subseteq \mathcal{T}$, define
 $\Psi^{s}(B,\mathcal{S})$:=
    \[
    \Phi^{s}_{\text{ext}}(\mathcal{S})\;\land\; 
    \Phi_{\text{prop}}(\mathcal{S}) \;\land
    \bigwedge_{f(t) \in \mathcal{S}} \bigwedge_{\substack{a \in \args(F,f)}} a\cmpsignstrict{\coeffsign{f}s}  \subst{t}{x}{\ivbound{B}{s}} 
    \]
    where $\Phi^{s}_{\text{ext}}(\mathcal{S})$ and $\Phi_{\text{prop}}(\mathcal{S})$ are as defined in \cref{def:reqpivot-param}.
\end{definition}

Thus, for any sign $s\in\{-,+\}$ and subset $\mathcal{S}\subseteq\mathcal{T}$, $\Psi^s(B,\mathcal{S})$ asserts that $\mathcal{S}$ satisfies the \req\ conditions for direction $s$ and \emph{additionally} satisfies the clash condition for interval $B$. This clash condition ensures that for every ground argument $a$ of $f$ appearing in $F$, the propagated cell $f(\subst{t}{x}{\ivbound{B}{s}})$ (with $f(t)\in\mathcal{S}$) is different from the ground cell $f(a)$.

\begin{theorem}[Interval Satisfiability]
\label{thm:interval-certificate-param}
Let $B$ be a non‑empty integer interval, and assume that
\[
  \Biggl(F \land \bigwedge_{x \in B} \subst{Q}{x}{\repr{x}}\Biggr) \;\land\;
  \Biggl(\bigvee_{\mathcal{S}\subseteq \mathcal{T}}\Psi^+(B,\mathcal{S})\Biggr)\;\land\;
  \Biggl(\bigvee_{\mathcal{S}'\subseteq \mathcal{T}}\Psi^-(B,\mathcal{S}')\Biggr)
\]
is satisfiable. Then  $F \land \bigwedge_{x \in B} \subst{Q}{x}{\repr{x}}\wedge \forall x\;.\; x\not\in B \Rightarrow Q$ has a satisfiability certificate, and hence $\phi$ is satisfiable.
\end{theorem}

The formula in the premise of \cref{thm:interval-certificate-param} belongs to the language of Presburger arithmetic (i.e., SMTLIB LIA), and hence its satisfiability can be algorithmically checked. Due to this, an interval $[b_{\min},b_{\max}]$ such that this formula is satisfiable is an algorithmically checkable certificate for the satisfiability of $\phi$. In order to make the algorithmic check more efficient, one can also include the cell interpretation necessary for showing $F \land \bigwedge_{x \in B} \subst{Q}{x}{\repr{x}}$ to be satisfiable, and the sets $\mathcal{S}$ and $\mathcal{S}'$.

Algorithmic implementations of the satisfiability propagators whose existence is ensured by the proof of \cref{thm:interval-certificate-param} can often be obtained from the witness functions that are 
a byproduct of quantifier elimination on the
propagability conditions~$\Phi_{\text{prop}}(\mathcal{S})$ and~$\Phi_{\text{prop}}(\mathcal{S}')$. These propagators  can then be used in the recursive evaluation function $val$ defined in \cref{sec:certificate} to compute values of the function that a model assigns to an uninterpreted function symbol.

The theorem immediately yields an algorithm that checks satisfiability by searching for the interval~$[b_{\min},b_{\max}]$. This can be seen in \cref{alg:interval-check-param-withBigvee}. An extension of the algorithm that in the case of the return value \textsc{Sat} also returns a corresponding certificate is straightforward.

\begin{algorithm}
\caption{Interval‑Based Satisfiability Check}
\label{alg:interval-check-param-withBigvee}
\begin{algorithmic}[1]
\Require $\phi = F \land \forall x.\, Q$ in the fragment
\Ensure \textsc{Sat} or \textsc{Unsat}
\State $b_{\min} \gets 0$, $b_{\max} \gets 0$ 
\While{$F \land \bigwedge_{x=b_{\min}}^{b_{\max}} Q[x]$ is satisfiable}
    \State $\text{up}\gets\bigvee_{\mathcal{S}\subseteq \mathcal{T}}\Psi^+([b_{\min},b_{\max}],\mathcal{S})$
    \State $\text{down}\gets\bigvee_{{\mathcal{S}'\subseteq \mathcal{T}}}\Psi^-([b_{\min},b_{\max}],\mathcal{S}')$ 
    \If{$\left(F \land \bigwedge_{x=b_{\min}}^{b_{\max}} Q[x]\right) \land \text{up} \land \text{down}$ is satisfiable}
        \State \Return \textsc{Sat} 
    \EndIf
    \State $b_{\min} \gets b_{\min} - 1$, $b_{\max} \gets b_{\max} + 1$ \Comment{Expand interval symmetrically}
\EndWhile
\State \Return \textsc{Unsat}
\end{algorithmic}
\end{algorithm}

The correctness of the result \textsc{Sat} of the algorithm is an immediate consequence of \cref{alg:interval-check-param-withBigvee}. The correctness of the result \textsc{Unsat} follows from the fact that the termination condition of the while loop is a logical consequence of the input formula. However, it might also happen that the algorithm runs forever. We can exclude this under certain conditions:

\begin{theorem}[Relative Completeness]
\label{thm:completeness-bigvee}
If the input formula~$\phi$ is satisfiable and it satisfies the \req condition, then \cref{alg:interval-check-param-withBigvee} terminates and returns \textsc{Sat}.
\end{theorem}

While proofs of unsatisfiability are not the main goal of this paper, \Cref{alg:interval-check-param-withBigvee} can also detect unsatisfiability, in some cases.

\begin{theorem}[Termination for Finitely Unsatisfiable Formulas]
\label{thm:finiteunsatThenUNSAT}
If there exists a finite set $Z\subseteq\mathbb{Z}$ such that $F \land \bigwedge_{z\in Z} Q[x \mapsto \repr{z}]$ is unsatisfiable, then \cref{alg:interval-check-param-withBigvee} will terminate and return \textsc{Unsat}.
\end{theorem}

However, there are formulas that are not satisfiable in a standard model, for which the algorithm does not terminate.

\begin{example}[Formula Requiring Induction]
Consider the formula:
\[
f(0) = 0 \land f(c) \neq 0 \land \forall x.\, f(x) = f(x+1)
\]
This formula is unsatisfiable in standard integer arithmetic (applying induction in both directions yields $\forall n.\, f(n) = 0$, contradicting $f(c) \neq 0$). However, any finite set of instantiations is satisfiable by choosing $c$ to be a sufficiently large integer. Thus, \cref{alg:interval-check-param-withBigvee} will expand the interval indefinitely without detecting unsatisfiability.
\end{example}

\section{Propositional Encoding}

The algorithm presented so far relies on large disjunctions over all possible subset choices $\mathcal{S}\subseteq\mathcal{T}$ and $\mathcal{S}'\subseteq\mathcal{T}$. While theoretically elegant, these disjunctions lead to formulas that grow exponentially with the size of the set $\mathcal{T}$, making them computationally impractical. We now present a compact CNF encoding that avoids this issue.

Instead of enumerating all possible subset choices, we introduce auxiliary Boolean variables that represent membership in the subsets $\mathcal{S}$ and $\mathcal{S}'$. For each term $f(t)\in\mathcal{T}$, we create a propositional variable 
$p_{f(t)}^+$ that is true iff $f(t) \in \mathcal{S}$, and another variable~$p_{f(t)}^-$ that is true iff $f(t) \in \mathcal{S}'$.

We then encode the extremal condition $\Phi^s_{\text{ext}}$, the propagability condition $\Phi_{\text{prop}}$, and the clash conditions using these variables, yielding a compact representation of $\Psi^{s,\text{e}}(B)$ for each sign $s$.

\begin{definition}[CNF Encoding of Interval Extension Formulas]
\label{def:interval-extension-formulas-CNF}
For an interval $B = [b_{\min}, b_{\max}]$ and a sign~$s\in\{+,-\}$, define:
\[\Psi^{s,\text{e}}(B) \;\equiv\; \Phi^{s,\text{e}}_{\text{ext}} \;\land\; \Phi_{\text{prop}}^{\text{e}} \;\land\; \text{ClashEnc}^s\]
where:
\begin{itemize}
    \item Extremal condition encoding $\Phi^{s,\text{e}}_{\text{ext}}$:
    \begin{align*}
        \subst{
        \bigwedge_{f \in \mathcal{F}}\ \bigwedge_{\substack{f(t),\, f(t') \in \mathcal{T}}}
        \Bigl( (p^s_{f(t)} \land p^s_{f(t')}) \Rightarrow t = t' \;\land\; 
        (p^s_{f(t)} \land \neg p^s_{f(t')}) \Rightarrow t' \;\cmpsignstrict{\coeffsign{f}s}\; t \Bigr)}{x}{\repr{0}}
    \end{align*}
   
    \item Propagation encoding $\Phi_{\text{prop}}^{\text{e}}$:%
    \[
    \bigwedge_{\mathcal{T}_{\!-}\subseteq\mathcal{T}}
    \left(\bigwedge_{f(t)\in\mathcal{T}_{\!-}} \neg p^s_{f(t)} \;\wedge\; \bigwedge_{f(t)\in\mathcal{T}\setminus\mathcal{T}_{\!-}} p^s_{f(t)}\;\Rightarrow\; \forall x\, \forall \bar{u}\, \exists \bar{v}\; Q^{\mathcal{T}_{\!-}}\right),
    \]
    where $\bar{u}=(u_{f(t)})_{f(t)\in\mathcal{T}_{\!-}}$ and $\bar{v}=(v_{f(t)})_{f(t)\in\mathcal{T}\setminus\mathcal{T}_{\!-}}$, and $Q^{\mathcal{T}_{\!-}}$ is obtained from $Q$ by simultaneously replacing each $f(t)\in\mathcal{T}_{\!-}$ by $u_{f(t)}$ and each $f(t)\in\mathcal{T}\setminus\mathcal{T}_{\!-}$ by $v_{f(t)}$.

    \item Clash condition encoding $\text{ClashEnc}^s$:
\[  \bigwedge_{f(t)\in\mathcal {T}}\bigwedge_{a \in \args(F,f)} \left(p^s_{f(t)} \Rightarrow a\cmpsignstrict{\coeffsign{f}s}  \subst{t}{x}{\ivbound{B}{s}}\right)    \]
\end{itemize}
\end{definition}

Note that if for some $f\in\mathcal{{F}}$ the ground parts of its argument terms (obtained by substituting $x=0$) are pairwise different, then at most one $p^s_{f(t)}$ can be true in any satisfying assignment of $\Psi^{s,\text{e}}(B)$. This follows from the extremal condition, which would otherwise require two distinct ground terms to be equal.

\begin{property}[Encoding Correctness]
\label{lem:encoding-correctness}
For any interval $B = [b_{\min}, b_{\max}]$,  and sign~$s\in\{+, -\}$, 
 $\Psi^{s,\text{e}}(B)$ is satisfiable iff there exists $\mathcal{S}\subseteq\mathcal{T}$ such that $\Psi^{s}(B,\mathcal{S})$ holds.
\end{property}

\section{Experiments}

Our experiments test the conjecture that the algorithm introduced in the previous two sections is able to improve
upon state-of-the-art $\smt$ solvers concerning the capability of showing formulas in the considered class to be satisfiable.
We prepared a test suite consisting of 28 handcrafted satisfiable
problems\footnote{All the data for our experiments are available at our public repository \\\url{https://github.com/MarekDanco/f-definitions}}
that, in addition, satisfy the \req condition. All problems in our set contain up to two constants, one or two function
symbols of arity 1, simple ground constraints and one universally quantified assertion.
Two examples of the problems in our set are
\begin{equation}\label{eq:exp1}
	f(4)=7 \:\land\: \forall x.\: f(x+1)=f(x)+1
\end{equation}  %
and
\begin{equation}\label{eq:exp2}
	d > 1\:\land\: f(0)=0 \:\land\: g(2) > f(d+1) \land\forall x \:.\: f(x-1)=f(x)+g(x)+g(x+d).
\end{equation}
To simplify implementation, we used versions that have an additional guard to restrict the range of the universally quantified variable to the naturals instead of the integers. The resulting formulas are of the form $F\land (\forall x \:.\: 0\leq x \Rightarrow Q)$. To study solver behavior over increasing model size, we also used versions where the guard restricts the range to a finite interval with upper bound $10^c$, with  $c\in \{1,\dots,5\}$. The resulting formulas are of the form  $F\land \forall x \:.\: (0\leq x \leq 10^c\Rightarrow Q)$.

Apart from the upper bound $10^c$, our problems do not contain big integers in function arguments, which would necessitate smarter initialization of the interval~$[b_{\min},b_{\max}]$ to avoid numerous loop iterations of the while loop in \cref{alg:interval-check-param-withBigvee}.

We implemented our algorithm using the \textsc{z3} Python API. Our implementation checks only upward induction using~$\Psi^+$, omitting downward induction, which is sound due to the guard described in the previous paragraph. It turned out that our algorithm terminates on all problems in negligible time.
Moreover, the value of the parameter~$c$ does not noticeably influence runtime
for the bounded versions of the test problems. We found the former observation surprising,
while the latter was to be expected, due to the use of induction.

We also applied the \textsc{smt} solvers \textsc{z3} and \textsc{cvc5}
directly to our problems, explicitly enabling model-based quantifier instantiation~\cite{ge-cav09}.
Additionally, we ran \textsc{cvc5} in \textsc{sygus} inference mode.
To study the behavior of direct application of the three $\smt$ solvers to our problems, and the influence
of the parameter~$c$ on their runtime, we ran our experiments with a 10-minute time limit.
When setting $c = 1$, the \textsc{smt} solvers terminate on our
problems immediately. On problem~\eqref{eq:exp1} with $c = 3$, \textsc{cvc5} takes
54 seconds to return \textsc{sat}, and for $c \geq 4$, the solver times out.
With $c \geq 2$, \textsc{z3} returns \textsc{unknown} after 12 seconds. For problem~\eqref{eq:exp2},
both solvers fail as $c$ increases: \textsc{cvc5} times out when $c \geq 3$,
and \textsc{z3} returns \textsc{unknown} for $c \geq 2$ (typically within 44–50 seconds).
Figure~\ref{fig:scaling} depicts how the two solvers perform on our problems with increasing $c$.
We also tested the performance of the three solvers on unbounded problems, that is with the
bound $0 \le x \le 10^c$ removed. \textsc{z3} solved 4 of those problems while \textsc{cvc5}
solved none.

\textsc{cvc5} in \textsc{sygus} mode solved 15 of our problems regardless
of the asserted bounds.
One problem that \textsc{cvc5} in \textsc{sygus} mode did not solve is
\begin{align*}
	f(0) = 1 \;\land \:\forall x \:. \:f(x+1) = 2*f(x).
\end{align*}
It is reasonable to assume that the function $\lambda x. \:2^{x}$ is not in \textsc{cvc5}
default \textsc{sygus} grammar. Further, Figure~\ref{fig:scaling} clearly shows
that the \textsc{sygus} approach is fast and consistent, but only on problems where
the functions conform to the pre-defined grammar.

We further tested the performance of the three solvers within a 1-minute time limit
(denoted by the -min suffix in Figure~\ref{fig:scaling}).
The performance of \textsc{z3} was the same. \textsc{cvc5} lost one solution for
$c = 2$ and $c = 3$. Additionally, \textsc{cvc5} in \textsc{sygus} mode lost one
solution for the unbounded problems.

\begin{figure}[ht]
	\centering
	\includegraphics[width=0.6\linewidth]{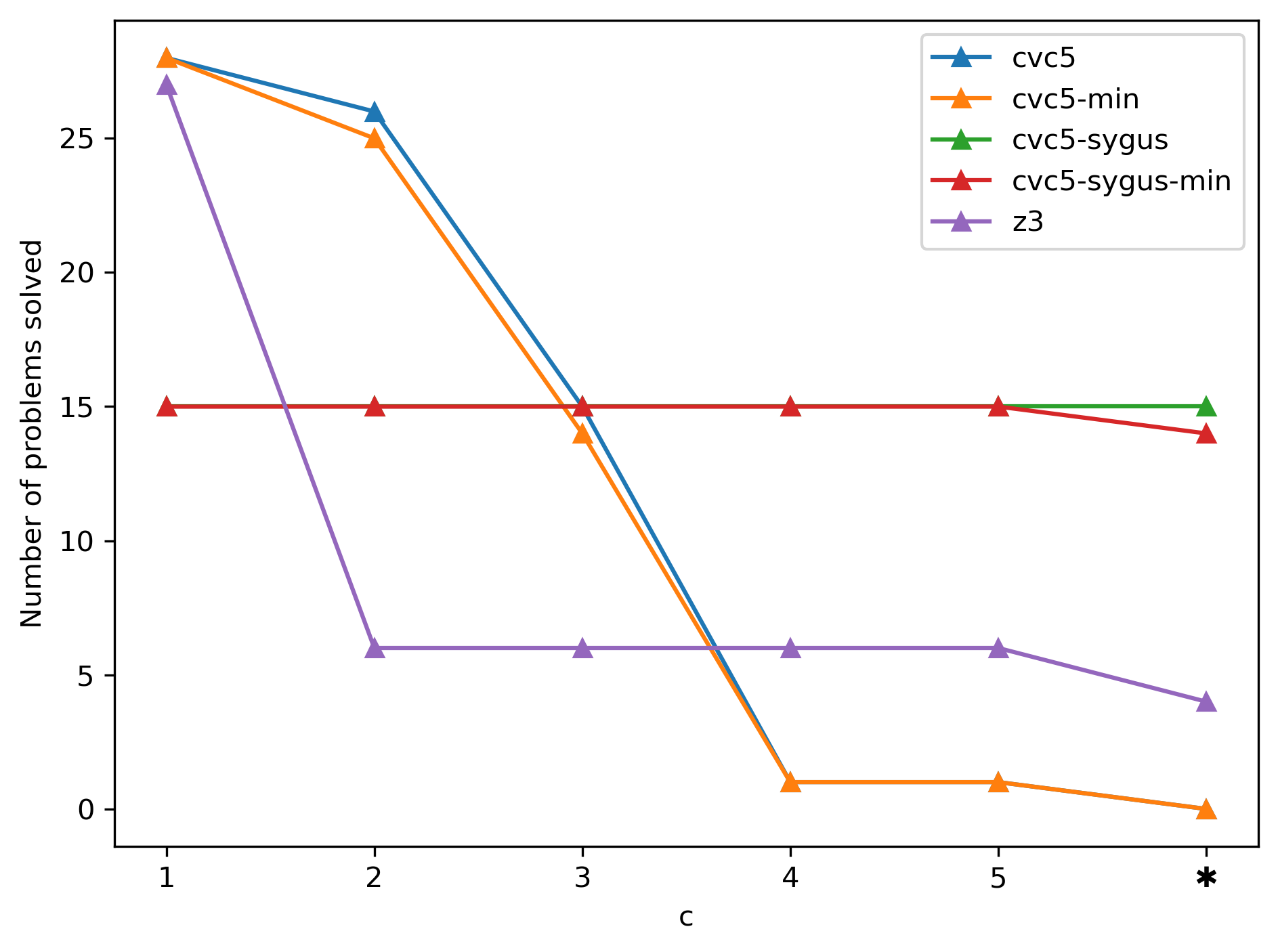}
	\caption{Problems solved with increasing $c$. $*$ refers to unbounded problems.}
	\label{fig:scaling}
\end{figure}

\section{Conclusion}
\label{sec:conclusion}

We have presented a method for certifying satisfiability of formulas with universal quantifiers and uninterpreted function symbols, and a corresponding algorithm for a specific class of formulas in the SMTLIB theory UFLIA. By relying on induction, the method succeeds for formulas without finite models and for formulas satisfiable only by functions lacking closed-form representations.
Beyond the syntactic restrictions of the input fragment, the algorithm also requires the semantic $\req$ condition to enable inductive extension. In future work, we aim to relax this restriction, for instance by employing
stronger forms of induction. Our long term vision is a general algorithm that is parametric in the underlying theories associated with the domain and codomain of the uninterpreted function symbols, and that depends only on a predicate defining the order in which to apply induction.

\bibliographystyle{splncs04}

\appendix

\section{Appendix}

\subsection*{Proof of Property~\ref{prop:recmonot}:}

  Assume a cell $u$ with $\mathit{inst}(u)\neq\bot$, a  cell~$v\in\Gamma_{I,\mathit{inst}(u)}(\phi)\setminus \mathit{def}(I)\setminus X_z$ with $\mathit{inst}(v)\neq\bot$, and assume that $\mathit{inst}(v)\succeq \mathit{inst}(u)$. Then both $v\in X_{\mathit{inst}(v)}$ and $v \in \Gamma_{I,\mathit{inst}(u)}(\phi)$, and hence $X_{\mathit{inst}(v)}\cap \Gamma_{I,\mathit{inst}(u)}(\phi)\neq\emptyset$. 
  This contradicts the requirement that $X_{\mathit{inst}(v)}\cap \mathit{def}(I)=\emptyset$ and $X_{\mathit{inst}(v)}\cap\bigcup_{z'\prec \mathit{inst}(v)} \Gamma_{I,z'}(\phi)=\emptyset$.
  \qedhere

\subsection*{Proof of \cref{thm:interval-certificate-param}:}

  Suppose the conjunction is satisfiable, let $I^*$ be an interpretation satisfying it, and let $\mathcal{S}$ be such that $I^*\models \Psi^+(B,\mathcal{S})$, and $\mathcal{S}'$ such that $I^*\models \Psi^-(B,\mathcal{S}')$. Assume that  $B = [b_{\min}, b_{\max}]$. We construct a certificate as in \cref{def:certificate}:

\begin{enumerate}
    \item \textbf{Pre‑satisfiability certificate:} Define the cell interpretation~$I$ as follows:
      \begin{itemize}
        \item For every uninterpreted constant $c$ in $\phi$, set $I(c)=I^*(c)$.
        \item For every $a\in args(F \land \bigwedge_{x \in B} \subst{Q}{x}{\repr{x}}, f)$,  where $f$ is an uninterpreted function symbol, $I(f(I^*(a)))=I^*(f)(I^{*}(a))$
    \end{itemize}
    By construction, $I$ is finite and satisfies $F \land \bigwedge_{x \in B} \subst{Q}{x}{\repr{x}}$ because $I^*$ satisfies that formula and $I$ agrees with $I^*$ on all relevant cells.
       
    \item \textbf{Well‑order $\preceq$ on $\mathbb{Z}$:}
    \[0, 1, -1, 2, -2, 3, -3, \dots, b_{\max}, b_{\min},b_{\max}+1, b_{\max}+2, \dots,b_{\min}-1, b_{\min}-2, \dots\]		
    
  \item \textbf{Propagated cells $X_z$ for $z \in \mathbb{Z}$:}
    \[
    X_z = 
    \begin{cases}
    \emptyset & \text{if } z \in B, \\[4pt]
    \bigl\{ f(I(\subst{t}{x}{\hat{z}})) \mid f(t)\in\mathcal{S} \bigr\} & \text{if } z > b_{\max}, \\[4pt]
    \bigl\{ f(I(\subst{t}{x}{\hat{z}})) \mid f(t)\in\mathcal{S}' \bigr\} & \text{if } z < b_{\min}.
    \end{cases}
    \]
    Note that the extremal conditions $\Phi^+_{\text{ext}}(\mathcal{S})$ and $\Phi^-_{\text{ext}}(\mathcal{S}')$ ensure that for $z\not\in B$, $f\in\mathcal{F}$, $|\{ n \mid f(n)\in X_z\}|\leq 1$.

    \item \textbf{Satisfiability propagators:}
    \begin{itemize}
        \item For $z \in B$: In this case, $\subst{(x\not\in B \Rightarrow Q)}{x}{\hat{z}}$ is a tautology and the propagator is trivial.
        \item For $z\not\in B$: the condition $\Phi_{\text{prop}}(\mathcal{S})$ (if $z > b_{\max}$) or $\Phi_{\text{prop}}(\mathcal{S}')$ (if $z < b_{\min}$) provides, for any given values for all non‑propagated cells in $\Gamma_{I,z}(\phi)$ and each $f(t)\in\mathcal{S}$, values $v_{f(t)}$  such that assigning these values to each corresponding cells in $X_z$ satisfies~$\subst{Q}{x}{\repr{z}}$.
    \end{itemize}
\end{enumerate}
It remains to verify the disjointness conditions of \cref{def:certificate}:
\[X_z\cap \mathit{def}(I)=\emptyset \text{ and } X_z\cap\bigcup_{z'\prec z} \Gamma_{I,z'}(\phi)=\emptyset.\]
For $z \in B$, $X_z = \emptyset$, and hence both conditions hold trivially.  For $z > b_{\max}$, any element of $X_z$ has the form $f(I(\subst{t}{x}{\repr{z}}))$ with $f(t) \in \mathcal{S}$.
    \begin{itemize}
    \item To prove that this element of $X_z$ cannot be in $\mathit{def}(I)$, let $f'(n)$ be an arbitrary, but fixed element of $\mathit{def}(I)$ with non-zero arity. We prove that $f'(n)\neq f(I(\subst{t}{x}{\repr{z}}))$. If $f$ is a different function symbol than $f'$ this is certainly the case, and hence it suffices to prove $f(n)\neq f(I(\subst{t}{x}{\repr{z}}))$, which makes it necessary to prove $I(n)\neq I(\subst{t}{x}{\repr{z}})$. Due to the construction of the cell interpretation~$I$, $n=I(a)$, with $a\in args(F \land \bigwedge_{x \in B} \subst{Q}{x}{\repr{x}}, f)$.
      \begin{itemize}
      \item       If $a\in args(F, f)$, the clash condition $a\cmpsignstrict{\coeffsign{f}s}\subst{t}{x}{\ivbound{B}{s}}$ in $\Psi^+(B,\mathcal{S})$ guarantees that for every ground argument $a \in \args(F,f)$:
        \[
        I\models a\cmpsign{\coeffsign{f}+}  \subst{t}{x}{\repr{b}_{\max}} \cmpsignstrict{\coeffsign{f}+} \subst{t}{x}{\repr{z}},
        \]
        where the strict inequality follows because $t$ is linear, and $z > b_{\max}$. Hence $I(a)\neq I(\subst{t}{x}{\repr{z}})$.
      \item If  $a\not\in args(F, f)$, then there is a $z_B\in B$ such that $a\in args(Q[x\leftarrow \hat{z}_B],f)$.   Now, since $z_B\leq b_{\max}<z$,  again  $I\models a\cmpsign{\coeffsign{f}+}  \subst{t}{x}{\repr{b}_{\max}}$, and  $I(a)\neq I(\subst{t}{x}{\repr{z}})$.
      \end{itemize}

      \item To prove that this element of $X_z$ cannot appear in any $\Gamma_{I,z'}(\phi)$ with $z' \prec z$, let $f'$ be an uninterpreted function symbol, $t'\in args(Q,f')$,  and prove that $f'(I(\subst{t'}{x}{\repr{z}'}))\neq f(I(\subst{t}{x}{\repr{z}}))$. Again it suffices to analyze the case where $f$ and $f'$ are the same function symbols, and to prove that
        $I(\subst{t'}{x}{\repr{z}'})\neq I(\subst{t}{x}{\repr{z}})$. Observe that since $t'$ and $t$ are arguments of the same function symbol, the respective coefficients of $x$ in $t$ and $t'$ are the same. Also observe that $z'\prec z$ and $b_{\max}<z$ implies that $z'<z$. Due to this,
        \[
          I\models \subst{t'}{x}{\hat{z}'}\cmpsignstrict{\coeffsign{f}+} \subst{t'}{x}{\hat{z}}.
        \]
        Moreover, $\Phi^{+}_{\text{ext}}(\mathcal{S})$ ensures 
                \[
          I\models \subst{t'}{x}{\hat{z}}\cmpsignstrict{\coeffsign{f}+}\subst{t}{x}{\hat{z}}.
        \]
Hence $I(\subst{t'}{x}{\repr{z}'})\cmpsignstrict{\coeffsign{f}+} I(\subst{t}{x}{\repr{z}})$, which implies that the two sides of the inequality cannot be the same.
    \end{itemize}

  The case $z < b_{\min}$ is symmetric, using $\Psi^-(B,\mathcal{S}')$. Note that the fact that the upward propagation block of elements greater than $b_{\max}$ occurs before the downward propagation block of elements less than $b_{\min}$ in the order $\preceq$ does not create a problem, since  propagation is still restricted to cells not occurring before.

Thus all certificate requirements are satisfied. By \cref{thm:certThenSat}, $\phi$ is satisfiable.
\qedhere

\subsection*{Proof of \cref{thm:completeness-bigvee}:}

  Assume $\phi$ is satisfiable and satisfies \req condition for $(\mathcal{S}^*,\mathcal{S}'^*)$. Let $J$ be an interpretation with $J\models\phi$. For each function $f \in \mathcal{F}$, let $c_f$ be the common coefficient of $x$ in all its arguments, and define:
\[
v_{\max}^f = \max\{\, J(\subst{t}{x}{\repr{0}}) \mid f(t) \in \mathcal{T}\,\}, \qquad
v_{\min}^f = \min\{\, J(\subst{t}{x}{\repr{0}}) \mid f(t) \in \mathcal{T} \,\}.
\]
The extremal conditions $\Phi^+_{\text{ext}}(\mathcal{S}^*)$ and $\Phi^-_{\text{ext}}(\mathcal{S}'^*)$ ensure that
\begin{itemize}
    \item For all $f(t) \in \mathcal{S}^*$: $J(\subst{t}{x}{\repr{0}}) = v_{\max}^f$, and
    \item For all $f(t) \in \mathcal{S}'^*$: $J(\subst{t}{x}{\repr{0}}) = v_{\min}^f$.
\end{itemize}
Now choose
\begin{minipage}{\linewidth-2cm}
\vspace*{-0.5cm}  
\begin{align*}
  b_{\max} &= \max\left\{\, \left\lceil \frac{J(a) - v_{\max}^f}{c_f} \right\rceil \;\middle|\; a \in \args(F,f), f \in \mathcal{F} \,\right\}\\
b_{\min} &= \min\left\{\, \left\lfloor \frac{J(a) - v_{\min}^f}{c_f} \right\rfloor \;\middle|\; a \in \args(F,f), f \in \mathcal{F} \,\right\},
\end{align*}
\end{minipage}
and set $B = [b_{\min}, b_{\max}]$. We verify that $J$ satisfies:
\begin{enumerate}
    \item \textbf{Base formula}: $J \models F \land \bigwedge_{x \in B} Q[x]$, since $J \models \phi$.
    
    \item \textbf{Upward propagation condition}: 
    For the specific $\mathcal{S}^*$, we have:
    \begin{itemize}
        \item $\Phi_{\text{ext}}^+(\mathcal{S}^*)$ holds by definition,
        \item $\Phi_{\text{prop}}(\mathcal{S}^*)$ holds by the \req condition,
        \item For each $f(t)\in\mathcal{S}^*$ and $a \in \args(F,f)$:\\ By construction of $b_{\max}$ and since $J(\subst{t}{x}{\repr{0}})=v_{\max}^f$, we have:   %
        \[
        J(a) \cmpsign{\coeffsign{f}+} c_f b_{\max} + v_{\max}^f = c_f  b_{\max} + J(\subst{t}{x}{\repr{0}}) = J(\subst{t}{x}{\hat{b}_{\max}}),
        \]
        which is exactly the clash condition in $\Psi^+(B,\mathcal{S}^*)$.
    \end{itemize}
    Hence $J \models \Psi^+(B,\mathcal{S}^*)$.
    
    \item \textbf{Downward propagation condition}: 
    Similarly, for $\mathcal{S}'^*$ we have:
    \begin{itemize}
        \item $\Phi_{\text{ext}}^-(\mathcal{S}'^*)$ holds,
        \item $\Phi_{\text{prop}}(\mathcal{S}'^*)$ holds by the \req condition,
        \item For each $f(t)\in\mathcal{S}'^{*}$ and $a \in \args(F,f)$:
        by construction of $b_{\min}$,
        \[
        J(a)\cmpsign{\coeffsign{f}-}  c_f b_{\min} + v_{\min}^f = c_f b_{\min} + J(\subst{t}{x}{\repr{0}}) = J(\subst{t}{x}{\repr{b}_{\min}}),
        \]
        which is the clash condition in $\Psi^-(B,\mathcal{S}'^*)$.       
    \end{itemize}
    Thus $J \models \Psi^-(B,\mathcal{S}'^*)$.
\end{enumerate}

Therefore, $J$ satisfies
\[
F \land \bigwedge_{x \in B} Q[x] \;\land\; \Psi^+(B,\mathcal{S}^*) \;\land\; \Psi^-(B,\mathcal{S}'^*),
\]
and consequently also satisfies the disjunctions $\bigvee_{\mathcal{S}}\Psi^+(B,\mathcal{S})$ and $\bigvee_{\mathcal{S}'}\Psi^-(B,\mathcal{S}')$.

Let $N = \max(|b_{\min}|, |b_{\max}|)$. \cref{alg:interval-check-param-withBigvee} expands the interval symmetrically from $[0,0]$. When it reaches iteration $n = N$, the interval becomes $[-N, N]$, which contains $B$. At this iteration, line 5 will check the satisfiability of:
\[
F \land \bigwedge_{x=-N}^{N} Q[x] \;\land\; \bigvee_{\mathcal{S}}\Psi^+([-N,N],\mathcal{S}) \;\land\; \bigvee_{\mathcal{S}'}\Psi^-([-N,N],\mathcal{S}').
\]
Since $[-N, N]\supseteq B$, $J$ also satisfies this formula. Hence, the algorithm terminates with \textsc{Sat}.
\qedhere

\end{document}